\documentclass[preprint,5p,authoryear]{elsarticle}
\usepackage{indentfirst}
\usepackage{amsmath,amssymb}
\usepackage{graphicx,subfigure,epstopdf}
\usepackage{epic}
\usepackage{subfigure}
\usepackage{graphicx}
\usepackage{epsfig}
\usepackage{arydshln}
\usepackage{enumerate}
\usepackage[english]{babel}
\usepackage{fancyhdr}
\usepackage{lastpage}
\usepackage{color}
\usepackage{dsfont}
\usepackage{tikz}
\usepackage{caption}
\usepackage{changes}
\usepackage{booktabs}
\usepackage{threeparttable}
\usepackage{animate}
\usepackage{natbib}
\usepackage{caption}
\usepackage[hyperfootnotes=false]{hyperref}
\usepackage{todonotes}
\hypersetup{
colorlinks=true,
linkcolor=black,
citecolor=cyan
}
\usetikzlibrary{positioning}

\setcitestyle{longnamesfirst,round,authoryear,semicolon,sort&compress}
\newcommand{\EQ}{\begin{eqnarray}}
\newcommand{\EN}{\end{eqnarray}}
\newcommand{\EQQ}{\begin{eqnarray*}}
\newcommand{\ENN}{\end{eqnarray*}}
\newcommand{\col}{\mbox{col}}

\newtheorem{thm}{Theorem}

\newtheorem{lem}{Lemma}
\newtheorem{rem}{Remark}
\newtheorem{defi}{Definition}

\newtheorem{prob}{Problem}
\newtheorem{exam}{Example}

\newenvironment{proof}{\noindent{\em Proof:\/}}{\hfill $\Box$\par}

\newcommand{\myr}{}

\newcommand{\meng}{}
\allowdisplaybreaks
\begin{document}
\begin{sloppypar}
\begin{frontmatter}
\title{Learning nonlinear dynamics in synchronization of knowledge-based leader-following networks\tnoteref{footnoteinfo}}
%
\tnotetext[t1]{This work was supported by the Natural Sciences and Engineering Research Council of Canada, and National Natural Science Foundation of China under Project 61773322.
This paper has not been presented at any conference.
\newline
* Corresponding author.}

\author[Wang]{Shimin Wang}\ead{shimin.wang@queensu.ca}
\author[meng]{Xiangyu Meng}\ead{xmeng5@lsu.edu}
\author[Zhang]{Hongwei Zhang*}\ead{hwzhang@hit.edu.cn}
\author[Frank]{Frank L. Lewis}\ead{lewis@uta.edu}
\address[Wang]{Department of Chemical Engineering, Queen's University, Kingston, Ontario, K7L 3N6, Canada.}
\address[meng]{Division of Electrical and Computer Engineering, Louisiana State University, Baton Rouge, LA 70803, USA.}
\address[Zhang]{School of Mechanical Engineering and Automation, Harbin Institute of Technology, Shenzhen, Guangdong 518055, P.R. China}
\address[Frank]{UTA Research Institute, University of Texas at Arlington, Ft. Worth, TX 76118, USA}

\begin{abstract}

Knowledge-based leader-following synchronization  of heterogeneous nonlinear multi-agent systems is a  challenging problem, since the leader's dynamic information is unknown to {\myr any follower node}. This paper proposes a learning-based fully distributed observer for a class of nonlinear leader systems, which can simultaneously learn the leader's dynamics and states. {\myr This} class of leader dynamics {\myr is rather general and} does not require a bounded Jacobian matrix. Based on this learning-based distributed observer, we further synthesize an adaptive distributed control law for solving the leader-following synchronization problem of multiple Euler-Lagrange systems subject to an uncertain nonlinear leader system. The results are illustrated by a simulation example.
\end{abstract}

\end{frontmatter}

\section{Introduction}
Synchronization of networked multi-agent systems has found its applications in various scenarios, such as flocking \citep*{bullo2009distributed}, coupled distributed estimation \citep*{olfati2012coupled}, and formation of multiple robots \citep*{dorfler2010geometric,feng2019finite,roy2021adaptive}, and has been extensively studied for the past several decades.
A central class of synchronization problems are  the so-called leader-following synchronization problems (also called cooperative
tracking problems \citep{zhang2012adaptive}), which aim to drive each follower system to behave in the same motion or oscillate concurrently with the leader system \citep*{chen2014pattern,isidori2014robust,radenkovic2018distributed}.

According to \cite{wang2006theoretical}, leader-following synchronization problems can be divided into two categories: power-based leader-following synchronization problems and knowledge-based leader-following synchronization problems.
For the former one, both followers and the leader node share the same system dynamics. For the latter one, the leader's dynamic knowledge (i.e., parameters) is unaccessible to any follower and each follower adaptively learns the leader's dynamic knowledge  through neighborhood {\myr interactions}. In this sense, the leader is also called an uncertain leader in the literature \citep*{wanghuang2018,baldi2020distributed}.

Most recent works \citep*{wu2017adaptive,klotz2014asymptotic} solved the leader-following synchronization problem by assuming that some/all followers know the leader's parameters. In contrast, solving knowledge-based leader-following synchronization problem is much challenging, since no followers can get access to the leader's knowledge (i.e., parameters) and only part of them can sense the leader's state or output signals. There are generally two approaches to handle the uncertain leader. One approach is the so-called distributed adaptive internal model design approach \citep*{su2013cooperative}. Then, under the assumption that the leader has a unitary relative degree and strongly minimum phase properties, the synchronization problem can be converted to a global adaptive robust stabilization problem.
Another is the adaptive distributed observer design approach, which was studied recently in \cite{wanghuang2018} and \cite{wangmeng2020}, where each follower node maintains an observer, estimating the state and/or learning the parameters of the leader node in a distributed manner, and then a decentralized controller is designed for each follower node to drive it to track the output of its observer.
Some early tries of the observer design can also be found in \cite{modares2016optimal}, where the convergence analysis of the estimated parameters {\myr is not provided.} Very recently, \cite{baldi2020distributed} also studied the knowledge-based leader-following synchronization problem and proposed an adaptive distributed observer to {\myr simultaneously} estimate the state and learn the knowledge of a linear leader. They pointed out that it is a starting point to extend this approach to more general cases. Although much attention has been paid to uncertain linear leader dynamics, insufficient investigation has been reported for knowledge-based leader-following networks with  nonlinear dynamics, to the best of the {\myr authors'} knowledge.

However, in practical applications, dynamics of the leader are often described by  nonlinear systems  \citep*{wangslotine2004,wang2006theoretical,zhang2012adaptive}, and the nonlinear leader system may contain some parameters unknown to any of the followers.
It is noted that the techniques developed in \cite{wanghuang2018} and \cite{baldi2020distributed} for {\myr linear leaders with unknown parameters}
cannot be applied to nonlinear systems. {\myr In this
paper, we consider an uncertain nonlinear leader}. Compared with its linear counterpart, the knowledge-based leader-following synchronization problem with a nonlinear leader is recognized to be much challenging, and is worthy of further  investigation.

Note that by considering the multi-agent system consisting only of the leader node and the learning-based observers, the distributed observer design problem is essentially the same with the knowledge-based leader-following synchronization problem \citep*{wang2006theoretical} and has the well-known frequency estimation problem \citep*{HsuOrtega1999TAC-globally} as a special case. 
In the literature, synchronization of coupled nonlinear systems can be achieved under either of the following three common conditions: the global Lipschitz condition \citep*{wangslotine2004,zhou2006adaptive,yu2009pinning}, the quadratic condition \citep*{lu2004synchronization,delellis2010synchronization}, and the contracting condition \citep*{wang2006theoretical}. The relationship among these conditions was reported in \cite{delellis2010quad}.
Essentially, the aforementioned works try to use a sufficiently large constant to stabilize the coupled nonlinear dynamic networks.
For more general scenarios, when the nonlinear dynamics of the leader do not satisfy the above conditions, a practical approach for synchronization of networked nonlinear systems is to use neural networks to approximate the unknown nonlinearities \citep*{zhang2012adaptive}, which only guarantees boundedness of the synchronization error, without estimating the leader's dynamics.

{\myr It is worth noting that, in the literature of multi-agent systems, many works require the knowledge of the whole  communication graph, such as its Laplacian matrix, which is a certain global information. This destroys  the scalability of the design in some extent.
For example, to handle a
nonlinear leader with known leader's parameters, \cite{DongChen202021} assume that the term related to the nonlinearity
and the network topology can be bounded by
some smooth function. This approach relies on the global
knowledge of the communication graph, which is not fully
distributed according to \cite{li2014designing}.
}

Motivated by the above-mentioned statements, in this paper, we consider a more practical scenario of knowledge-based leader-following synchronization problems, where the leader node has {\meng a nonlinear dynamical model which} is not known by any followers. We aim at designing a fully distributed control law, {\meng which does not require any global
knowledge of the communication graph. Therefore, the
design is not affected by the topology change of the communication
graph. In other words, it is {\myr totally} scalable.} We first establish a learning-based fully distributed observer for a class of nonlinear leader systems whose {\meng nonlinearities can be written as the product between a weight vector and
a regressor matrix}. Under some standard assumptions, this learning-based distributed observer can estimate and pass the leader's  state to each follower through local {\myr interactions} without knowing the leader's  parameters. {\meng If} the leader's regressor matrix is persistently exciting, this distributed observer can also asymptotically adaptively learn the leader's parameters.
{\myr Since the nonlinear dynamics of the leader and information
of the communication network gets involved in the analysis and design of learning-based fully distributed observer, the design of learning-based fully distributed observer proposed in \cite{wanghuang2018} and \cite{baldi2020distributed} cannot be used to estimate and learn the state and knowledge of a nonlinear leader.}
Our observer design removes the bounded Jacobian matrices assumption or the quadratic condition on the coupled nonlinear systems \citep*{wangslotine2004,wang2006theoretical,zhou2006adaptive}.
 Moreover, the global information of the communication graph has been completely removed by using an adaptive technique. {\myr It is different from the one proposed in \cite{li2014designing} and \cite{DongChen202021}}, which is not applicable in our case of nonlinear leader dynamics.
As an application, based on this learning-based distributed observer, we further synthesize an adaptive distributed control law for solving the leader-following synchronization problem of multiple Euler - Lagrange systems subject to an uncertain nonlinear leader system.

{\meng The contributions of this article are: (1) relaxing the
bounded Jacobian matrices assumption or the quadratic
condition on the coupled nonlinear systems \citep{wangslotine2004,wang2006theoretical,zhou2006adaptive} and (2) using
an adaptive technique to obtain fully distributed design
without relying on any global information of the communication
graph.}
%


The rest of this paper is organized as follows: Section \ref{section1}   formulates the problem. In Section \ref{section2}, we establish a
learning-based distributed observer for a nonlinear leader system whose parameters are unknown to any follower nodes. As an application of the learning-based distributed observer, we apply the learning-based distributed observer to synthesize an adaptive distributed control law for solving
the leader-following synchronization problem of multiple Euler - Lagrange systems subject to an uncertain nonlinear leader system in Section \ref{section4}.
A simulation example is given in Section \ref{section5},  and Section \ref{section6} concludes the paper.

\textbf{Notation:} Let $\otimes$ denote the Kronecker product of matrices. {\myr A function with $k$ continuous derivatives is called a $C^k$ function. $\mathds{1}_N=\col(1,\dots,1)$, and $I_n$ denotes the identity matrix with dimension $n$. $\mathds{R}^{+}$ denotes all the positive real numbers}. For $X_1,\dots,X_k\in \mathds{R}^{n\times m}$, $\mbox{col}(X_1,\dots,X_k)=[X_{1}^{T}, \dots,X_{k}^{T}]^T$ and $$\mbox{blkdiag} (X_1,\dots,X_k)=\left[
                                                      \begin{array}{ccc}
                                                        X_1 &   &   \\
                                                          & \ddots &   \\
                                                          &   & X_k\\
                                                      \end{array}
                                                    \right].
$$ For  $ x \in \mathds{R}^m$, unless indicated otherwise, $x_i$ denotes the $i$-th component of $x$, $\|x\|$ and $\|x\|_p$ denote the $2$-norm and the $p$-norm of $x$, respectively.

\section{Preliminaries and problem formulation}\label{section1}
\subsection{Graph theory}\label{sectionA}
The communication/interaction network of a multi-agent system composed of $N$ followers and one leader is described by a graph $\bar{\mathcal{G}}=\left(\bar{\mathcal{V}},\bar{\mathcal{E}}\right)$ with
$\bar{\mathcal{V}}=\{1,\dots,N,N+1\}$ and $\bar{\mathcal{E}}\subseteq \bar{\mathcal{V}}\times \bar{\mathcal{V}}$ being the node set and the edge set, respectively.
Here node $N+1$ is associated
with the leader and nodes $i$, $i = 1,\dots,N$,
are associated with the followers. For $i=1,\dots,N, N+1$ and $j=1,\dots,N$, $(i,j) \in \bar{\mathcal{E}}$ if and only if node $j$ can get the
information of node $i$ for control purpose. Let
$\bar{\mathcal{N}}_i=\{j|(j,i)\in \bar{\mathcal{E}}\}$
denote the neighbor set of agent $i$.
Let ${\cal G}=({\cal V},{\cal E})$  denote
the induced subgraph of $\bar{\mathcal {G}}$ with ${\cal
V}=\{1,\dots,N\}$, which captures the interaction among follower nodes. Assume that $\mathcal{\bar{G}}$ contains a spanning tree with node $N+1$ being the root, and $\mathcal{G}$ is an undirected graph. Let $\mathcal{\bar{L}}$ be the
Laplacian matrix of the graph $\mathcal{\bar{G}}$, and $H$ is obtained by deleting the last row and column of $\mathcal{\bar{L}}$. Then, $H$ is a symmetric
positive definite matrix \citep*{zhang2012adaptive}. More details of the graph theory can be found in \cite{godsil2013algebraic}.

\subsection{Problem formulation}
Let the follower nodes be described by general nonlinear systems
\begin{align}\label{followerdynamic}
    \dot{x}_i=& \textbf{f}_i(x_i,\tau_i),\nonumber\\
y_i=&\textbf{h}_i(x_i,\tau_i),\;i=1,\dots,N,
\end{align}
where $x_i\in \mathds{R}^{n_{i}}$, $y_i\in \mathds{R}^{n}$ and $\tau_i\in \mathds{R}^{m_{i}}$ are the state, measurement output and control input of the $i$th follower, and $\textbf{f}_i(\cdot)$ and $\textbf{h}_i(\cdot)$ are globally defined and sufficiently smooth functions vanishing at the origin.

The leader's output signal $q_{N+1} \in \mathds{R}^{n}$, is generated by the following nonlinear system
\begin{subequations}\label{leader}
\begin{align}
\dot{v}&=p(v,\omega),\label{leadera} \\
  q_{N+1}&=Ev, \label{leaderb}
\end{align}
\end{subequations}
where $v\in \mathds{R}^{m}$, $p(\cdot,\omega)$ is a globally defined and sufficiently smooth function vanishing at the origin, $\omega\in \mathds{R}^{l}$ is a constant vector consisting of the unknown parameters of the leader node, $E\in \mathds{R}^{n\times m}$ is a known constant matrix.  The uncertain nonlinear system \eqref{leader} is assumed to satisfy
\begin{align}\label{pphi}{ p(v, \omega)}=\phi\left(v\right)\omega,\end{align}
where the regressor matrix $\phi\left(\cdot\right)\in \mathds{R}^{m\times l}$ is known. It is assumed that given any compact set $\mathds{V}_0$, there exists a compact set $\mathds{V}$ such that, for any $v(0)\in \mathds{V}_0$, $v(t)\in \mathds{V}$ for all $t\geq 0$. The system \eqref{leader} encompasses a class of nonlinear systems that can generate stable limit cycles, such as the well-known Van der Pol system, {\myr which can describe lots of periodic behavior arising from physics, biology, chemistry and  engineering \citep*{buchli2006engineering}}.

\begin{figure*}[htp]
\centering
\includegraphics[clip, trim=140  588  63 118, scale=0.9]{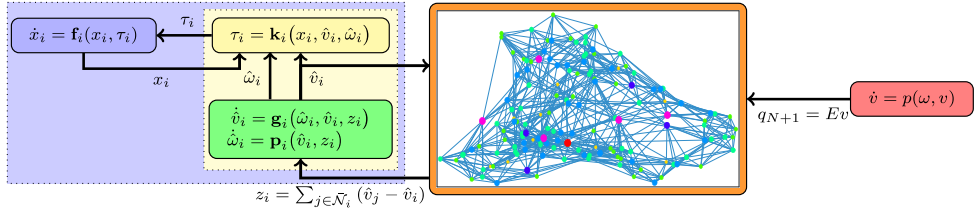}
\caption{{\myr Schematic of learning-based distributed control design}}\label{figschematic}
\end{figure*}
{\myr The design schematic is shown in Figure \ref{figschematic}. We consider the following class of distributed control laws}
\begin{subequations}\label{generalconrol}\begin{align}
\tau_i=&\textbf{k}_i(x_i,\hat{v}_i,\hat{\omega}_i),\\
\dot{\hat{v}}_i=&\textbf{g}_i\left(\hat{\omega}_i,\hat{v}_i,\sum\nolimits_{j\in\bar{\mathcal{N}}_i}(\hat{v}_j-\hat{v}_i)\right),\label{generalconrolb}\\\
\dot{\hat{\omega}}_i=&\textbf{p}_i\left(\hat{v}_i,\sum\nolimits_{j\in\bar{\mathcal{N}}_i}(\hat{v}_j-\hat{v}_i)\right),\label{generalconrolc}
\end{align}\end{subequations}
where $\hat{v}_{N+1}=v$, and for $i=1\dots,N$, $\hat{v}_i\in \mathds{R}^{m}$ and $\hat{\omega}_i\in \mathds{R}^{l}$ are the estimation of $v$ and $\omega$, respectively; $\textbf{k}_i(\cdot)$, $\textbf{g}_i(\cdot)$ and $\textbf{p}_i(\cdot)$ are sufficiently smooth functions vanishing at the origin.

Now we are ready to formulate the knowledge-based leader-following synchronization (KLFS) problem for a class of nonlinear multi-agent systems.

\begin{prob}[KLFS Problem]\label{problFS}
Consider the multi-agent system (\ref{followerdynamic}) and (\ref{leader}), and the corresponding communication graph $\mathcal{\bar{G}}$. Given any compact set $\mathds{V}_0\subset \mathds{R}^m$ containing the origin, design a distributed control law in the form of \eqref{generalconrol}, such that for any initial condition $x_i(0)\in\mathds{R}^{n_i}$, $\hat{v}_i(0)\in\mathds{R}^{m}$, $\hat{\omega}_i(0)\in\mathds{R}^{l}$, and $v(0)\in \mathds{V}_0$,
$x_i(t)$ is bounded for all $t\geq0$, and
$$\lim\limits_{t\rightarrow\infty}\left(y_i\left(t\right)-q_{N+1}\left(t\right)\right)=0, \quad i=1,2,\dots,N.$$
\end{prob}

Note that for the special case that  $v$ and $\omega$ are known to all followers, the control law \eqref{generalconrol} will reduce to the following form
\begin{align}\label{contrl}
\tau_i=\textbf{k}_i(x_i,{v},\omega).
\end{align}
Then the KLFS problem will reduce to the traditional adaptive nonlinear output regulation problem, which has been well studied in \cite{huang2004nonlinear} and \cite{isidori2014robust}.

In practice, the leader's dynamic knowledge (i.e., parameters) may be unknown to all followers, and only a few followers who have direct links to the leader can sense the leader's state or output information.
Then, in order to solve the KLFS problem, we need to design a learning-based distributed observer for each follower node to estimate the state and simultaneously learn the parameters of the leader node. The learning-based distributed observer, which is the main kernel in solving Problem \ref{problFS}, is defined as follows.
%

\begin{defi}[Learning-based distributed
observer] A dynamic of the form \eqref{generalconrolb} and \eqref{generalconrolc} is called a learning-based distributed observer for the leader \eqref{leader} if, given a digraph $\mathcal{\bar{G}}$ and any compact set $\mathds{V}_0\subset \mathds{R}^m$ containing the origin, there exit global defined and sufficiently smooth functions $\textbf{g}_i(\cdot)$ and $\textbf{p}_i(\cdot)$, such that for any initial condition $\hat{\omega}_i(0)\in\mathds{R}^{l}$, $\hat{v}_i(0)\in\mathds{R}^{m}$ and $v(0)\in \mathds{V}_0$,
$\hat{v}_i(t)$ is bounded for all $t\geq0$, and
$$\lim\limits_{t\rightarrow\infty}\left(\hat{v}_i\left(t\right)-v\left(t\right)\right)=0.$$
 Further, under certain conditions, the observer \eqref{generalconrolc} can adaptively learn the actual values of the knowledge-based leader in the sense that
$$\lim\limits_{t\rightarrow\infty}\left(\hat{\omega}_i\left(t\right)-\omega\right)=0.$$
\end{defi}

Before proceeding, we recall the definition of persistent excitation {\myr and two useful lemmas, which will be used in the sequel.}

\begin{defi}\citep*{Anderson} A bounded piecewise continuous function $f: [0,+\infty)\mapsto \mathds{R}^{n\times m}$ is said to be persistently exciting if there exist positive constants $\epsilon$, $t_0$, and $T_0$ such that,
$$\frac{1}{T_0}\int^{t+T_0}_{t} f(s) f^T (s) ds\geq\epsilon I_n,~~~~\forall t\geq t_0.$$
\end{defi}
%


{\myr \begin{lem}\citep*{chen2015stabilization}\label{chen2015stabilization1}
\textbf{(i)} For any continuous function $f(x,d):\mathds{R}^n\times \mathds{R}^l\mapsto \mathds{R}$, there exist smooth functions $a(x), b(d)\geq 0$, such that $|f(x,d)|\leq a(x)b(d)$.\\
\textbf{(ii)} For any continuous function $f(x,d):\mathds{R}^n\times \mathds{R}^l\mapsto \mathds{R}$ satisfying $f(0,d)=0$, there exist smooth functions $m(x,d)\geq 0$, such that $|f(x,d)|\leq m(x,d)\|x\|$.
\end{lem}
\begin{lem}\citep*{huang2004nonlinear}\label{huang2004nonlinear1} Let $f:\mathds{R}^m\times \mathds{R}^n\times \mathds{R}^p\mapsto \mathds{R}$ be a $C^1$ function satisfying $f(0,0,d)=0$ for all $d\in \mathds{D}$, with $\mathds{D}$ being a compact set of $\mathds{R}^p$. Then, there exist smooth functions $F_1:\mathds{R}^m\mapsto \mathds{R}$ and $F_2:\mathds{R}^n\mapsto \mathds{R}$ satisfying $F_1(0)=0$ and $F_2(0)=0$ such that
$$|f(x,y,d)|\leq F_1(x)+F_2(y),\forall x\in \mathds{R}^{m}, y\in \mathds{R}^{n}, d\in \mathds{D}.$$
\end{lem}
}

\section{Learning-based distributed observer design}\label{section2}
 This section devotes to the design and analysis of the learning-based distributed observers. Then in Section \ref{section4}, we will apply this observer to solve a leader-following synchronization problem of multiple Euler-Largrange systems.

\subsection{Learning-based distributed observers design}
Note that the parameter vector $\omega$ in the leader system \eqref{leader} is constant, and thus bounded.
To estimate both states and parameters of the leader node, we design the following learning-based distributed observer for node $i$,
\begin{subequations}\label{compensator2}
\begin{align}
\dot{\hat{v}}_i &=p(\hat{v}_i,\hat{\omega}_i)  + \hat{\kappa}_i\rho_i\left( z_i\right) z_i, \label{compensator2b}\\
\dot{\hat{\omega}}_{i} &=\mu\phi^T\left(\hat{v}_i\right) z_i, \label{compensator2c}\\
\dot{\hat{\kappa}}_i &= \rho_i\left( z_i\right) z_i^T z_i, \label{compensator2d}\\
 z_i &= \sum \nolimits_{j \in \mathcal{\bar{N}}_i} (\hat{v}_j-\hat{v}_i),\;i=1,\dots,N,\label{ndiff}
\end{align}
\end{subequations}%
where $\hat{v}_{N+1}=v$, $\hat{v}_i\in \mathds{R}^m$ is the estimation of $v$, $\hat{\omega}_i\in \mathds{R}^{l}$ is the estimation of $\omega$, $\hat{\kappa}_i\in \mathds{R}$ is a dynamic gain to {\meng relax the dependence on an} unknown bound $\bar{\kappa}_i$ induced by the matrix $H$, uniform bound of state $v$ and unknown parameter $\omega$. {\meng Here}
$$p(\hat{v}_i,\hat{\omega}_i)=\phi \left( \hat{v}_i \right)\hat{\omega}_i,$$
where $\phi \left( \cdot \right) \in \mathds{R}^{m \times l}$ is defined in \eqref{pphi}, $\rho_i\left(\cdot\right)\geq 1$ is a smooth positive function to be designed and {\myr $\mu$ can be selected as any positive scalar}.

\begin{rem}
A salient feature of the observer \eqref{compensator2} is its fully distributed nature, and only those follower nodes who have direct links from the leader need to get access to the leader's state information, instead of its parameters $\omega$.
%
%
Closely related results can be found in \cite{wangslotine2004,wang2006theoretical} and \cite{zhou2006adaptive}. However, they use a fixed linear gain to stabilize the coupled nonlinear dynamic networks under assumptions that each coupled nonlinear agent must satisfy the global bounded Jacobian matrix assumption or the quadratic condition proposed in \cite{lu2004synchronization}.

Moreover, the observer design \eqref{compensator2} is fully distributed in the sense that the coupled global information $\bar{\kappa}$ induced by the communication graph and knowledge-based leader system has been completely removed by using an adaptive technique, i.e., equation \eqref{compensator2d}. Note that this technique is different from methods proposed in \cite{li2014designing} and \cite{sun2021distributed}, {\meng in which the existence of time-varying coupling weights is based on the assumption of the global Lipschitz condition and the quadratic condition, respectively. However, such
conditions are not imposed in this paper. Therefore, the methods in \cite{li2014designing} and \cite{sun2021distributed} are not applicable in our case.}
\end{rem}

\begin{rem}
 The choice of $\rho_i\left( \cdot\right)$ depends on the structure of $p(\cdot,\omega)$, where $\omega$ is unknown.
 For a class of polynomial differential systems \eqref{leader}, with $p(v,\omega)$ being real polynomials of degree $m_0\geq1$,
$$\|p(v,\omega)\|^2\leq\sum\nolimits_{k=0}^{2m_0}c_k(\omega)\|v\|^{k}$$
where $c_k(\omega)$ is some unknown constant induced by the uncertain parameter $\omega$. Then
 $\rho_i\left( z_i\right)$ {\meng can be chosen as}
$$\rho_i(z_i) = a_i \sum\nolimits_{k=0}^{2m_0-2}\|z_i\|^{k} +b_i,$$
where $a_i$ and $b_i$ are \added{\myr any} positive numbers. 
\end{rem}

\subsection{Convergence analysis}
Define the estimation errors as
$\tilde{{v}}_i=\hat{v}_i-v$, $\tilde{\omega}_i=\hat{\omega}_i-\omega$, and $\tilde{\kappa}_i=\hat{\kappa}_i-\bar{\kappa}_i$, respectively. Their time derivatives  along the trajectories \eqref{leader} and \eqref{compensator2} are
\begin{subequations}\label{reeq12}
\begin{align}
  \dot{\tilde{{v}}}_i =& \bar{\kappa}_i \rho_i(z_{i}) z_{i}+\phi\left(\hat{v}_i\right)\tilde{\omega}_i\nonumber \\
 & + \tilde{\kappa}_i \rho_i(z_{i}) z_{i}+\left(\phi\left(\hat{v}_i\right)-\phi\left(v\right)\right)\omega,\label{reeq12a}\\
\dot{\tilde{\omega}}_{i} = &\mu\phi^T\left(\hat{v}_i\right)z_{i},\label{reeq12b}\\
\dot{\tilde{\kappa}}_i =& \rho_i\left( z_i\right) z_i^T z_i,\;i=1,\dots,N. \label{compensator3c}
\end{align}%
\end{subequations}%
Define $\tilde{{v}}=\col\left(\tilde{{v}}_1,\dots,\tilde{{v}}_N\right)$, $\hat{v} =\col\left(\hat{v}_1,\dots,\hat{v}_N\right)$, $\tilde{{\kappa}}_d=\mbox{blkdiag}\left(\tilde{{\kappa}}_1,\dots,\tilde{{\kappa}}_N\right)$, $\phi_d (\hat{v})=\mbox{blkdiag}\left(\phi\left(\hat{v}_1\right),\dots,\phi\left(\hat{v}_N\right)\right)$,  $\bar{{\kappa}}_d=\mbox{blkdiag}\left(\bar{{\kappa}}_1,\dots,\bar{{\kappa}}_N\right)$, $\rho_d=\mbox{blkdiag}\left(\rho_1,\dots,\rho_N\right)$, and
$z=\col(z_{1},\dots,z_{N})$.
%
Then, we have
\begin{equation} \label{eqev}
z =-(H\otimes I_m)\tilde{{v}},
\end{equation}
where $H$  is defined in Section~\ref{sectionA}.
Equations \eqref{reeq12a} and \eqref{reeq12b} can be rewritten in the following compact form
\begin{subequations}\label{compensator3}
\begin{align}
  \dot{\tilde{{v}}} =&( \bar{\kappa}_d\rho_d\otimes I_m) z +\phi_d (\hat{v}) \tilde{\omega}+(\tilde{\kappa}_d\rho_d\otimes I_m) z\nonumber\\
  &+\left[\phi_d (\hat{v})-I_N\otimes\phi\left(v\right)\right](\mathds{1}_N\otimes \omega), \label{compensator3a}\\
\dot{\tilde{\omega}} =&\mu\phi_d^T\left(\hat{v}\right)z.\label{compensator3b}
\end{align}
\end{subequations}
where $\mathds{1}_N\in2 \mathds{R}^N$ and $I_m \in \mathds{R}^{m\times m}$ denote the vector of all
ones and the identity matrix, respectively

Before analyzing the convergence of the proposed learning-based distributed observer, we first establish the following result.

\begin{lem}\label{lemmarho} Consider systems \eqref{leader} and \eqref{compensator2}.
\begin{enumerate}
  \item
There exit sufficiently smooth positive function $\gamma_i(z_i) $ and positive constants $\beta_M$ and $\lambda_M$ such that
\begin{align}\sum\limits_{i=1}^{N}\| p(\hat{v}_i,\omega)-p(v,\omega)\|^2 \leq \beta_M^2 \lambda_M\sum\limits_{i=1}^{N}{\gamma}_i(z_{i})\|z_{i}\|^2.\label{sprocedure_1}\end{align}
\item Moreover, if {system \eqref{leader}} is a class of polynomial differential system with degree $m_0$, then
$${\gamma}_i(z_{i})=\sum\nolimits_{k=0}^{2m_0-2}\|z_i\|^{k}.$$
\end{enumerate}
\end{lem}
\begin{proof} \textbf{1):} Define
\begin{align} f_i(\tilde{v}_i,v,\omega)=\| p(\tilde{v}_i+v,\omega)-p(v,\omega)\|,   \label{specur2}\end{align} which is a continuously differentiable function satisfying $f_i(0,v,\omega)=0$, for $i=1,\dots,N$.
{\myr By the second part of Lemma \ref{chen2015stabilization1}}, there exists a smooth function $m_i(\tilde{{v}}_i(t),v(t),\omega)\geq 0$ such that%
\begin{align}
    f_i(\tilde{v}_i(t),v(t),\omega)\leq m_i(\tilde{{v}}_i(t),v(t),\omega)\|\tilde{{v}}_i(t)\|,~~~~ \forall t \geq 0.\label{specur22}
\end{align}
For the continuous function $m_i(\tilde{{v}}_i(t),v(t),\omega)$, {\myr by the first part of Lemma~\ref{chen2015stabilization1}}, there exist
smooth functions $\alpha_i(\tilde{{v}}_i(t))\geq 0$ and $\beta_i(v(t),\omega)\geq 0$, such that
\begin{align}m_i(\tilde{{v}}_i(t),v(t),\omega)\leq {\alpha}_i(\tilde{{v}}_i(t))\times{\beta}_i(v(t),\omega),~~ \forall t \geq 0.\label{specur222}\end{align}
Since $v(t)$ is {\myr uniformly} bounded in $t$, and $\omega$ is an unknown constant vector, there exists a positive constant $\beta_M=\sup_{t\geq 0}{\beta}_i(v(t),\omega)$ such that ${\beta}_i(v(t),\omega) \leq \beta_M$, for all $t\geq 0$.
{\myr Define
\begin{align}
F(\tilde{v},v,\omega)\equiv &\sum\nolimits_{i=1}^{N} f_i^2(\tilde{v}_i,v,\omega)\label{specur2222}
\end{align}
Then, from equations \eqref{specur2}, \eqref{specur22}, \eqref{specur222} and \eqref{specur2222}, we have
\begin{align*}
F(\tilde{v},v,\omega) = &\sum\nolimits_{i=1}^{N} f_i^2(\tilde{v}_i,v,\omega)\\
\leq &\sum\nolimits_{i=1}^{N} m_i^2(\tilde{{v}}_i(t),v(t),\omega)\|\tilde{{v}}_i(t)\|^2\\
\leq &\sum\nolimits_{i=1}^{N}   {\alpha}_i^2(\tilde{{v}}_i(t))\times{\beta}_i^2(v(t),\omega)\|\tilde{{v}}_i(t)\|^2\\
\leq&\sum\nolimits_{i=1}^{N} \beta_M^2{\alpha}_i^2(\tilde{{v}}_i)\|\tilde{{v}}_i\|^2\equiv \tilde{g}(\tilde{{v}}).
\end{align*}
As $\tilde{{v}}(z)$ is a function of $z$ via the relationship
\begin{align*}\tilde{v}=-(H^{-1}\otimes I_m)z.
\end{align*}
Therefore,
$$\tilde{g}(\tilde{{v}})=\tilde{g}(-(H^{-1}\otimes I_m)z)\equiv {g}(z)$$
Since the smooth positive definite function $ g(\cdot)$ and its first derivative vanishes at the origin, by Lemma~\ref{huang2004nonlinear1}, there exist smooth functions $\gamma_i(z_{i})$ and  some positive constant $\lambda_M$, such that%
   \begin{align*}
   F(\tilde{v},v,\omega) \leq& g(z) \\
    \leq &\beta_M^2 \lambda_M\sum\nolimits_{i=1}^{N}{\gamma}_i(z_{i})\|z_{i}\|^2.
\end{align*}}
\textbf{2):} For $l=1,\dots,m$, $p_l(\cdot,\omega): \mathds{R}^{m} \rightarrow \mathds{R}$ are continuously differentiable functions and $p_l(\cdot,\omega)$ are a class of polynomial differential systems with the largest degree $m_0$. It is noted that $\tilde{v}_i=\hat{v}_i-v$, for $i=1,\dots,N$. Then, for any $\hat{v}_i$ and $v$, we have
\begin{align*}
    p_l(\hat{v}_i,\omega)-p_l(v,\omega)
     =&\Big[\int_{0}^{1}\frac{\partial p_l(x,\omega)}{\partial x}\Big|_{x=v+\vartheta_l\tilde{v}_i}d\vartheta_l \Big]\tilde{v}_i.
\end{align*}
Then, for $l=1,\dots,m$, we will have
$$\|p_l(\hat{v}_i,\omega)-p_l(v,\omega)\|\leq \Big\|\int_{0}^{1}\frac{\partial p_l(x,\omega)}{\partial x}\Big|_{x=v+\vartheta_l\tilde{v}_i}d\vartheta_l \Big\|\|\tilde{v}_i\|.$$
In addition, $p_l(x,\omega)$ is in real polynomials form with the largest degree $m_0$. Thus $$ \sum_{l=1}^{m}\Big\|\int_{0}^{1}\frac{\partial p_l(x,\omega)}{\partial x}\Big|_{x=v+\vartheta_l \tilde{v}_i}d\vartheta_l \Big\|\leq \sum\nolimits_{k=0}^{m_0-1}{\beta}_k(\omega,v)\|\tilde{v}_i\|^{k},$$
where $\beta_k(\omega,v(t))$ are some unknown positive smooth functions induced by the uncertain parameter $\omega$ and unknown state $v(t)$. Since $v(t)$ is uniform bounded in $t$, and $\omega$ is an unknown constant vector, there exists a positive constant $\beta_M=\sup_{t\geq 0}{\beta}_k(v(t),\omega)$ such that ${\beta}_k(v(t),\omega) \leq {\beta}_M$, for all $t\geq 0$.
Hence, we have
\begin{align*}
\|p(v,\omega)-p(\hat{v}_i,\omega)\|^2\leq  {\beta}_{M}^2\sum\nolimits_{k=0}^{2m_0-2}\|\tilde{v}_i\|^{k+2}.
\end{align*}
It is noted that
\begin{align}
 \|\tilde{v}\|=\|(H^{-1}\otimes I_m)z\|\leq\bar{h}\|z\|,\nonumber
\end{align}
where $\bar{h}=\|H^{-1}\|$, and  $$\|x\|_p\leq \|x\|\leq n^{\frac{1}{2}-\frac{1}{p}}\|x\|_p$$ for any $p\geq2$ and $x\in \mathds{R}^{n}$, which further implies
$$\|x\|_{p}^p\leq \|x\|^p\leq n^{\frac{p}{2}-1}\|x\|_p^{p}.$$
 Then, for any $k\geq2$, we have
\begin{align}
\sum\nolimits_{i=1}^{N}\|\tilde{v}_i\|^k\leq&\big(\sum\nolimits_{i=1}^{N}\|\tilde{v}_i\|^2\big)^{k/2}\nonumber\\
\leq&\bar{h}^k\big(\sum\nolimits_{i=1}^{N}\|z_i\|^2\big)^{k/2}\nonumber\\
=&\bar{h}^k\big(\|z\|\big)^{k}\leq N^{\frac{k}{2}-1}\bar{h}^k\sum\nolimits_{i=1}^{N}\|z_i\|^{k}.\nonumber
\end{align}
Thus
\begin{align*}
\sum_{i=1}^{N}\|p(v,\omega)-p(\hat{v}_i,\omega)\|^2\leq& \beta_{M}^2\sum_{i=1}^{N}\sum_{k=2}^{2m_0}N^{k-1}\bar{h}^{k+2}\|z_i\|^{k}.
\end{align*}
Then, by letting ${\gamma}_i(z_{i})=\sum\nolimits_{k=0}^{2m_0-2}\|z_i\|^{k}$, we have
\begin{align*}
\sum\nolimits_{i=1}^{N}\|p(v,\omega)-p(\hat{v}_i,\omega)\|^2
 \leq {\beta}_{M}^2\lambda_M\sum\nolimits_{i=1}^{N}{\gamma}_i(z_{i})\|z_{i}\|^2
\end{align*}
for some positive constant $\lambda_M=N^{k-1}\bar{h}^{k+2}$.
\end{proof}

\begin{thm}\label{bounded2} Consider systems \eqref{leader} and \eqref{compensator2}. There exist some sufficiently smooth positive functions $\rho_i(\cdot)\geq 1$ for all $i=1,2,\dots,N$, such that, for all $\mu>0$, $\hat{{v}}_i(0)\in \mathbb{V}_0$ and $\hat{\omega}_i(0)\in \mathbb{R}^{l}$, {\myr the signals} $\hat{v}_i(t)$, $\hat{\kappa}_i(t)$ and $\hat{\omega}_i(t) $ exist and are bounded for all $t\geq 0$ and satisfy
\EQ
&&\lim\limits_{t\rightarrow\infty} \tilde{{v}}_i (t) =0,  \label{eq101a}\\
&&\lim\limits_{t\rightarrow\infty}\dot{\tilde{\omega}}_i(t)=0,  \label{eq101b}\\
&&\lim\limits_{t\rightarrow\infty}\phi (\hat{v}_i(t)) \tilde{\omega}_i(t) =0.  \label{eq101c}
 \EN
Moreover, if $\phi^T(v)$ is {persistently exciting}, then
\EQ \label{tomega}
\lim\limits_{t\rightarrow\infty} \tilde{\omega}_i(t)=0.
\EN
\end{thm}
\begin{proof}
Consider the Lyapunov function candidate
\begin{eqnarray}\label{reeq3b}
V  =\frac{1}{2}\big[{\tilde{{v}}}^T\left(H\otimes I_m\right){\tilde{{v}}}+\mu^{-1}{\tilde{\omega}}^T{\tilde{\omega}}+\sum\nolimits_{i=1}^{N}{\tilde{\kappa}}_i^T{\tilde{\kappa}}_i\big],
\end{eqnarray}
where $H$ is defined in Section \ref{sectionA} and is symmetric positive definite.
Differentiating (\ref{reeq3b}) along the trajectory of \eqref{reeq12} gives
\begin{align}\label{reeq4b}
\dot{V}  =&\tilde{{v}}^T\left(H\otimes I_m\right)\dot{\tilde{{v}}}+\mu^{-1}{\tilde{\omega}}^T\dot{{\tilde{\omega}}}+\sum\nolimits_{i=1}^{N}{\tilde{\kappa}}_i^T\dot{\tilde{\kappa}}_i\nonumber\\
=&-z^T(\bar{\kappa}_d\rho_d\otimes I_m)z -z^T(\tilde{\kappa}_d\rho_d\otimes I_m)z\nonumber\\
&-z^T\left[\phi_d (\hat{v})-I_N\otimes\phi\left(v\right)\right](\mathds{1}_N\otimes \omega)\nonumber\\
&-z^T\phi_d (\hat{v}) \tilde{\omega}
+\mu^{-1}{\tilde{\omega}}^T\dot{{\tilde{\omega}}}+\sum\nolimits_{i=1}^{N}{\tilde{\kappa}}_i^T\dot{\tilde{\kappa}}_i\nonumber\\
= &-\sum\nolimits_{i=1}^N\Big[\bar{\kappa}_i\rho_i(z_i)\|z_{i}\|^2-z_{i}^T\big[\phi (\hat{v}_i)-\phi (v)\big]\omega\Big]\nonumber\\
&-\sum\nolimits_{i=1}^{N}\big[\tilde{\kappa}_i\rho_i(z_i) z^T_i z_i-{\tilde{\kappa}}_i^T\dot{\tilde{\kappa}}_i\big]\nonumber\\
=&-\sum\limits_{i=1}^N\Big[\bar{\kappa}_i\rho_i(z_i)\|z_{i}\|^2-z_{i}^T\big[p(\hat{v}_i,\omega)-p(v,\omega)\big]\Big].
\end{align}
By using the $\mathcal{S}$ procedure and the inequality \eqref{sprocedure_1} in Lemma~\ref{lemmarho}, the equation in \eqref{reeq4b} implies $\dot{V}\leq 0$ if and only if there exits a $\lambda\geq 0$ such that%
\begin{align*}
    \sum\nolimits_{i=1}^{N}
    \begin{bmatrix}
     -\bar{\kappa}_i \rho_i(z_i) &1/2 \\
     1/2 & 0
    \end{bmatrix} \leq \lambda \lambda_M\beta_M^2 \sum\nolimits_{i=1}^{N} \begin{bmatrix}
     -\gamma_i(z_i) &0 \\
     0 & 1
    \end{bmatrix} .
\end{align*}
A sufficient condition to make the above inequality hold is%
\begin{align}
    \begin{bmatrix}
     -\bar{\kappa}_i \rho_i(z_i) + \lambda \lambda_M\beta_M^2 \gamma_i(z_i) &1/2 \\
     1/2 & -\lambda{\myr \lambda_M}\beta_M^2
    \end{bmatrix} <  0.\label{sprocedure2}
\end{align}
Therefore, {for any sufficiently smooth positive function} $\rho_i(\cdot)$ with the form%
\begin{align}\label{rhoz}
 \rho_i(z_i) = a_i \gamma_i(z_i) +b_i\geq \lambda \lambda_M\beta_M^2\gamma_i(z_i)/\bar{\kappa}_i
\end{align}
where $a_i$ and $b_i$ {are any positive numbers}, there exist $\bar{\kappa}_i$ and $\lambda$ to make \eqref{sprocedure2} hold and%
\begin{align}\label{VodtZ}
\dot{V} \leq - \varepsilon \sum\nolimits_{i=1}^{N} z^T_i z_i.
\end{align}%
for some $\varepsilon>0$.
Then, $V(t)$ is bounded, which means ${\tilde{{v}}}(t)$, ${\tilde{\omega}}(t)$ and ${\tilde{\kappa}}(t)$ are bounded, for all $t\geq 0$, and $\lim\limits_{t\rightarrow\infty}V(t)$ exits and is finite. Since $\tilde{v}(t)$ is bounded, $z(t)$ is bounded from \eqref{eqev}, and $\hat{v}(t)$ is bounded, for all $t\geq 0$. From \eqref{sprocedure_1} and the smoothness of $\gamma_i$, $i=1,\dots,N$, we know that $\tilde{\phi}_d(t)$ is bounded, for all $t\geq 0$. Again, using the smoothness of $\gamma_i(z_i)$, $i=1,\dots,N$,  $\rho_d(z(t))$ and $\dot{\rho}_d(z(t))$ are all bounded from \eqref{rhoz} and the fact that $z(t)$ is bounded, for all $t\geq 0$. From  (\ref{compensator3}), $\dot{{\tilde{{v}}}}(t)$ is bounded because $\rho_d(t)$, $z(t)$, {\myr $\phi_d(\hat{v}(t))$}, ${\tilde{\omega}}(t)$, {\myr $\phi(v(t))$}, and ${\tilde{\kappa}_d}(t)$ are all bounded for all $t\geq 0$.
{\myr The time derivative of $\dot{V}$ is
\begin{align}\label{ddvt}
\ddot{V}=&-\sum\limits_{i=1}^N\Big[\bar{\kappa}_i\dot{\rho}_i(z_i)\|z_{i}\|^2-\dot{z}_{i}^T\big[p(\hat{v}_i,\omega)-p(v,\omega)\big]\nonumber\\
&+2\bar{\kappa}_i\rho_i(z_i)z_{i}^T\dot{z}_{i}-z_{i}^T\big[\dot{p}(\hat{v}_i,\omega)-\dot{p}(v,\omega)\big]\Big].
\end{align}
From  Lemma \ref{lemmarho} and the fact that $z$ is bounded for all $t\geq0$, we have that $\|p(\hat{v}_i,\omega)-p(v,\omega)\|$ is bounded for all $t\geq0$.  Furthermore, by part \textbf{(i)} of Lemma \ref{chen2015stabilization1},  we can easily conclude that $\|\dot{p}(\hat{v}_i,\omega)\|$ and $\|\dot{p}(v,\omega)\|$ are bounded for all $t\geq0$ by sufficiently smooth positive functions of $\hat{v}_i$ and $v$, respectively.
In addition, we have proven that all the items $z_i$, $\dot{z}_{i}$, $\dot{\rho}_i(z_i)$, $\hat{v}_i$ and $v$ in \eqref{ddvt} are bounded for all $t\geq0$. which further implies $\ddot{V}(t)$ is bounded from \eqref{ddvt}, for all $t\geq 0$.}

By the Lyapunov-Like Lemma in \cite{r18},
we have
 $\lim\limits_{t\rightarrow\infty}\dot{V}(t)=0$.
From \eqref{VodtZ}, we have $$-\dot{V} \geq \varepsilon z^T z\geq 0$$
which implies that $ \lim\limits_{t\rightarrow\infty}z(t)=0$. From \eqref{eqev} and \eqref{compensator3b}, we can further have \eqref{eq101a} and (\ref{eq101b}), respectively.

To show (\ref{eq101c}), differentiating  $\dot{\tilde{{v}}}$ gives,
\begin{align}\label{reeq6b}
 \ddot{\tilde{{v}}}= &  \left(\bar{\kappa}_d\rho_d\otimes I_m\right)\dot{z}+ \big(\bar{\kappa}_d\dot{\rho}_d\otimes I_m\big)z\nonumber\\
 &+{\myr \big[\dot{\phi}_d (\hat{v})-I_N\otimes\dot{\phi}\left(v\right)\big]}(\mathds{1}_N\otimes \omega)\nonumber\\
 &+\dot{\phi}_d (\hat{v}) \tilde{\omega}+\phi_d (\hat{v})\dot{ \tilde{\omega}}+(\dot{\tilde{\kappa}}_d\rho_d\otimes I_m) z\nonumber\\
 & +(\tilde{\kappa}_d\dot{\rho}_d\otimes I_m) z+(\tilde{\kappa}_d\rho_d\otimes I_m) \dot{z}.
 \end{align}
We have shown that $\rho(t)$ and $\phi(t)$ are smooth, and $z(t)$, $\dot{z}(t)$, $\tilde{\omega}(t)$, $\dot{\tilde{\omega}}(t)$, and ${\tilde{\kappa}}(t)$ are all bounded, for all $t \geq 0$. We can also show that $\dot{\tilde{\kappa}}(t)$ is bounded from \eqref{compensator3c}, for all $t \geq 0$. Thus, $\ddot{\tilde{{v}}}(t)$ is bounded from \eqref{reeq6b}, for all $t \geq 0$.
By the Barbalat's lemma, we have $\lim\limits_{t\rightarrow\infty}\dot{\tilde{{v}}}(t)=0$, which further implies $$\lim\limits_{t\rightarrow\infty}\big[\phi\left(\hat{v}_i(t)\right)\tilde{\omega}_i(t)+\left(\phi\left(\hat{v}_i(t)\right)-\phi\left(v(t)\right)\right)\omega\big]=0.$$ Besides,
  for $i=1,\dots,N$,
\begin{align*}
\lim\limits_{t\rightarrow\infty}&\left[\phi\left(\hat{v}_i(t)\right)-\phi\left(v(t)\right)\right]\omega =0,
\end{align*}
due to \eqref{eq101a}.
Then, $$\lim\limits_{t\rightarrow\infty}\phi\left(\hat{v}_i(t)\right)\tilde{\omega}_i(t)=0.$$  Thus,
(\ref{eq101c}) holds.
Eqs.~(\ref{eq101a}) and \eqref{eq101c} imply
\begin{align*}
 \lim\limits_{t\rightarrow\infty} \phi({\hat{v}}_i(t) - \tilde{{v}}_i(t))\tilde{\omega}_i(t)= \lim\limits_{t\rightarrow\infty} \phi(v(t))\tilde{\omega}_i(t)
 = 0.
\end{align*}
If $\phi^T(v)$ is persistently exciting, by Lemma 2.4 of \cite{chen2015stabilization},
we have \eqref{tomega} from \eqref{eq101b}.
 \end{proof}

 \begin{rem}\label{vsxi}
 As a result of {Theorem} \ref{bounded2}, it is trivial to show that
\begin{subequations}\begin{align}
   \lim\limits_{t\rightarrow\infty} \big(E\hat{v}_i(t)-q_{N+1}(t)\big)=& 0,  \label{conzero}\\
     \lim\limits_{t\rightarrow\infty}  \big(E\dot{\hat{v}}_i(t)-\dot{q}_{N+1}(t)\big)=&0,~~i=1,\dots,N.\label{dotconzero}
 \end{align}
\end{subequations}
\end{rem}

\begin{rem}
The learning-based distributed observer \eqref{compensator2} features itself in two aspects.
 Unlike the nonlinear distributed observers proposed in many literatures, such as \cite{LJHuang2018auto} and \cite{DongChen202021}, the observer \eqref{compensator2} is independent of the leader's dynamic knowledge, i.e., parameters $\omega$. Moreover, the observer \eqref{compensator2} does not depend on the eigenvalues of either the Laplacian matrix or the adjacency matrix of the communication graph, which is normally required in the literature of distributed observer design. Therefore, our proposed observer is fully distributed. In particular, the observer \eqref{compensator2} can adaptively learn the true values of the leader's parameters.
\end{rem}



  \begin{exam}\label{remwu}
This example shows that the well-known Van der Pol system \citep*{r18} satisfies the requirements of the leader node.
Suppose the leader's states are generated by a Van der Pol oscillator of the following form
\begin{equation}\label{Van0}
 \dot{v}=p(v,\omega)= \left[
             \begin{array}{c}
               av_2 \\
               -bv_1 + c\left(1-v_{1}^2\right)v_2\\
             \end{array}
           \right],
\end{equation}
where $v=\textnormal{\col}(v_1, v_2)$ and $\omega=\textnormal{\col}(a,b,c)$.
When $a$, $b$, and $c$ are positive constants, the Van der Pol system will have a stable limit cycle with periodic $T$ for any {\myr nonzero} initial condition, which implies there exists $t_0$ such that $v_{l}(t)=v_{l}(t+T), \forall t\geq t_0, \forall l=1,2$ \citep*{buchli2006engineering}.
Then, we have
\begin{subequations}
\begin{align}
    \int_{t}^{t+T}&v_{2}(s)ds=\int_{t}^{t+T}\frac{1}{a}\dot{v}_{1}(s)ds =0  ,\nonumber\\
     \int_{t}^{t+T}&v_1(s)v_2(s)\left(1-v_{1}^2(s)\right)ds=0.\nonumber
 \end{align}
\end{subequations}
For any $t\geq t_0$, both $v_{1}(t)$ and $v_{2}(t)$ are periodic functions with periodic $T$,
\begin{align*}
\int_{t}^{t+T}v_{1}^2(s) ds >& 0,\\
\int_{t}^{t+T}v_{2}^2(s) ds >& 0,\\
 \int_{t}^{t+T}v_1^2(s)v_2^2(s)\left(1-v_{1}^2(s)\right)^2ds  >& 0.
 \end{align*}
From \eqref{Van0}
we have
\begin{align*}
   \frac{\partial p(v, \omega)}{ \partial \omega} =& \left[
                     \begin{array}{ccc}
                        v_2 &  0 &  0 \\
                        0 &  -v_1 & \left(1-v_{1}^2\right)v_2  \\
                     \end{array}
                   \right]=\phi\left(v\right).
\end{align*}
Then, It can be verified that
\begin{align*}
  \phi^T\left(v\right) \phi\left(v\right)=\left[
                                            \begin{array}{ccc}
                                              v_2^2   &  0 &  0 \\
                                               0 &  v_1^2 & -v_1v_2\left(1-v_{1}^2\right)  \\
                                               0 &  -v_1v_2\left(1-v_{1}^2\right) &  v_2^2\left(1-v_{1}^2\right)^2 \\
                                            \end{array}
                                          \right].
\end{align*}
Hence, for any $t\geq t_0$, we will have
\begin{subequations}
\begin{align}
    &   \int_{t}^{t+T} \phi^T\left(v(s)\right) \phi\left(v(s)\right)ds \nonumber\\
    =  &  \left[
                                            \begin{array}{ccc}
                                              \int_{t}^{t+T}v_2^2ds   &  0 &  0 \\
                                               0 &   \int_{t}^{t+T}v_1^2ds & 0  \\
                                               0 &  0 &   \int_{t}^{t+T}v_2^2\left(1-v_{1}^2\right)^2ds \\
                                            \end{array}
                                          \right]\nonumber\\
     >& 0, \nonumber
\end{align}
\end{subequations}
which implies that $\phi^T\left(v(t)\right)$ is {persistently exciting}.
 \end{exam}

\section{An application to leader-following synchronization of multiple Euler-Lagrange  systems}\label{section4}
{\meng The Euler-Lagrange equations are often used to describe
the evolution of a mechanical system  subject to
holonomic constraints, such as the dynamics of robot manipulators} \citep*{lewis2003control,roy2021adaptive}.
In this section, we apply the proposed learning-based distributed observer to solve the leader-following synchronization problem of multiple Euler-Lagrange systems subject to an uncertain nonlinear leader system.

\subsection{Problem formulation}
Consider a group of $N$ follower systems described by the following Euler-Lagrange system
\begin{equation}\label{MARINEVESSEL1}
  M_i\left(q_i\right)\ddot{q}_i+C_i\left(q_i,\dot{q}_i\right)\dot{q}_i+G_i\left(q_i\right)=\tau_i,
\end{equation}
where, for $i=1,\dots,N$, $q_i\in \mathds{R}^n$ is the vector of generalized position vector, $M_i\left(q_i\right)\in \mathds{R}^{n\times n}$ is the symmetric positive definite inertia matrix, $C_i\left(q_i,\dot{q}_i\right)\in \mathds{R}^{n\times n}$ is the coriolis and centripetal forces, and $\tau_i\in \mathds{R}^{n}$ is the control torque.
For $i=1,\dots,N$,  system \eqref{MARINEVESSEL1} has the following three properties \citep*{lewis2003control}:
\begin{enumerate}
\item The inertia matrix $M_i\left(q_i\right)$ is symmetric and uniformly positive definite, i.e. that there exists positive constants $\alpha$ and $\beta$ such that, for all $q_i$%
\begin{align*}
\beta I\geq  M_i\left(q_i\right) \geq \alpha I,
\end{align*}%
where $I$ is the identity matrix with appropriate dimensions.
\item For all $a, \dot{a}\in \mathds{R}^n$,
\begin{align}\label{YMQ}M_i\left(q_i\right)a+C_i\left(q_i,\dot{q}_i\right)\dot{a}+& G_i\left(q_i\right)\nonumber\\
=&Y_i\left(q_i,\dot{q}_i,a,\dot{a}\right)\theta_i,\end{align}
  where $Y_i\left(q_i,\dot{q}_i,a,\dot{a}\right)\in \mathds{R}^{n\times p}$ is a known regressor matrix and $\theta_i\in \mathds{R}^{p}$ is a constant vector consisting of the uncertain parameters of (\ref{MARINEVESSEL1}).
\item \label{prop1}$\big(\dot{M}_i\left(q_i\right)-2C_i\left(q_i,\dot{q}_i\right)\big)$ is skew symmetric, $\forall q_i, \dot{q}_i$.
 \end{enumerate}
%
%

\begin{prob}\label{ldlesp}
Consider the multi-agent system  consisting of (\ref{leader}) and (\ref{MARINEVESSEL1}). Let $\mathds{V}_0\subset \mathds{R}^m$ be a compact set containing the origin. Design a distributed control law $\tau_i$,
such that for any initial condition $q_i(0)\in\mathds{R}^{n}$, $\dot{q}_i (0)\in\mathds{R}^{n}$, and $v(0)\in \mathds{V}_0$,
$q_i(t)$, and $\dot{q}_i(t)$ are bounded for all $t\geq0$, and
$$\lim\limits_{t\rightarrow\infty}\left(q_i\left(t\right)-q_{N+1}\left(t\right)\right)=0.$$
\end{prob}

\subsection{Distributed observer based control law}

Define $\tilde{q}_i= q_i-E\hat{v}_i$ as the ``leader output tracking error'',  where $E\hat{v}_i$ can be regarded as an estimate of leader's output $q_{N+1}$ in (\ref{leader}).The estimated ``leader output derivative''
\begin{align}
    \dot{\hat{q}}_i = E \phi (\hat{v}_i) \hat{\omega}_i-\alpha\tilde{q}_i \label{claw1}
\end{align}
is formed by shifting the estimated derivative of the output $E\phi (\hat{v}_i) \hat{\omega}_i$ according to the tracking error $\tilde{q}_i$, where $\hat{v}_i$ and $\hat{\omega}_i$ are defined in \eqref{compensator2}, and $\alpha>0$ is a positive scalar to be designed.
The vector $s_i$ defined as%
\begin{align}
    s_i = \dot{q}_i - \dot{\hat{q}}_i\label{claw2}
\end{align}
for $i=1,\dots,N$, can be interpreted as a ``leader output derivative error''.
Take the adaptive control law to be
\begin{subequations}\label{adapcontrolaw}
\begin{align}
  \tau_i=&-K_{i}s_i+Y_i\hat{\theta}_i,\label{dstureq11i}\\
   \dot{\hat{\theta}}_i =& -\Gamma_i Y_i^T s_i,\label{thetagamma}
\end{align}
\end{subequations}
where $\Gamma_i$ is a diagonal matrix with positive diagonal entries and $K_i$ is a positive definite matrix, $\hat{\theta}_i$ is the parameter estimates, for $i=1,\dots,N$.
The control law \eqref{adapcontrolaw} includes a ``feedforward'' term $Y_i\hat{\theta}_i$ .
\subsection{Stability analysis}
\begin{thm} Consider systems \eqref{leader}, \eqref{MARINEVESSEL1}, and a digraph $\bar{\mathcal G}$. For any initial condition $q_i(0)\in\mathds{R}^{n}$, $\dot{q}_i (0)\in\mathds{R}^{n}$, and $v(0)\in \mathds{V}_0$, Problem \ref{ldlesp} is solvable by the control law {\eqref{adapcontrolaw}}.
\end{thm}
\begin{proof}
Differentiating (\ref{claw1}) and \eqref{claw2}, and considering $\tilde{q}_i= q_i-E\hat{v}_i$ give, 
\begin{subequations}\label{dclaw}
\begin{align}
\ddot{\hat{q}}_{i}&=E\phi (\hat{v}_i) \dot{\omega}_i+E\phi(\dot{\hat{v}}_i) \hat{\omega}_i-\alpha\big(\dot{q}_i-E\dot{\hat{v}}_i\big),\label{dclaw1}\\
 \dot{s}_i&=\ddot{q}_i-\ddot{\hat{q}}_{i}. \label{dclaw2}
 \end{align}
\end{subequations}
By equation \eqref{YMQ}, there exists a known matrix $$Y_i=Y_i\big(q_i,\dot{q}_i, \ddot{\hat q}_{i}, \dot{\hat q}_{i} \big)$$ and an unknown constant vector $\theta_i$ such that,
\begin{equation}\label{claw0}
 Y_i\theta_i=M_i\left(q_i\right)\ddot{\hat q}_{i}+C_i\left(q_i,\dot{q}_i\right)\dot{\hat q}_{i}+G_i\left(q_i\right).
\end{equation}
Substituting $Y_i\theta_i$ into (\ref{MARINEVESSEL1}) gives
\begin{align}\label{undistur2}
  M_i\left(q_i\right)(\ddot{q}_i-\ddot{\hat{q}}_{i})+C_i\left(q_i,\dot{q}_i\right)(\dot{q}_i-\dot{ \hat q}_{i}) &\nonumber\\
+G_i\left(q_i\right)-G_i\left(q_i\right)+Y_i\theta_i&=\tau_i.
\end{align}
Considering (\ref{claw2}), \eqref{undistur2} can be written as
\begin{eqnarray}\label{undistur3}
M_i\left(q_i\right)\dot{s}_i=-C_i\left(q_i,\dot{q}_i\right)s_i+\tau_i-Y_i\theta_i.\label{undistur1i}
\end{eqnarray}
Substituting \eqref{adapcontrolaw} into (\ref{undistur3}) gives, for $i=1,\dots,N$,
 \begin{equation}\label{MARINEVESSEL14}
 M_i\left(q_i\right)\dot{s}_i=-C_i\left(q_i,\dot{q}_i\right)s_i-K_is_i+Y_i\tilde{\theta}_i
\end{equation}
where $\tilde{\theta}_i=\hat{\theta}_i-\theta_i$.
For $i=1,\dots,N$,  consider the following Lyapunov function candidate
\begin{equation}\label{MARINEVESSEL15}
V_i=\frac{1}{2}\big[s_{i}^TM_i\left(q_i\right)s_i+\tilde{\theta}_{i}^T\Gamma_i^{-1}\tilde{\theta}_{i}\big],
\end{equation}
The time derivative of \eqref{MARINEVESSEL15} along the trajectory (\ref{MARINEVESSEL14}) yields
\begin{align}\label{MARINEVESSEL16}
  \dot{V}_i&=s_{i}^TM_i\left(q_i\right)\dot{s}_i+\frac{1}{2}s_{i}^T\dot{M}_i\left(q_i\right)s_i+\tilde{\theta}_{i}^T\Gamma_i^{-1}\dot{\tilde{\theta}}_{i} \nonumber\\
           &=-s_{i}^TK_is_i \leq 0, ~i=1,\dots,N.
\end{align}%

Now since $V_i(t)\geq 0$ and $\dot{V}_i(t) \leq 0$, {\myr $V_i(t)$} remains bounded, for all $t\geq 0$. Equation \eqref{MARINEVESSEL15} further implies that  both $s_i(t)$ and $\tilde{\theta}_{i}(t)$ are bounded, for all $t\geq 0$. Furthermore, this in turn implies that $q_i(t)$, $\dot{q}_i(t)$, and $\hat{\theta}_i(t)$ are all
bounded, for all $t\geq 0$. Noticing the closed-loop dynamics in \eqref{MARINEVESSEL14}, this shows that $\dot{s}_i(t)$ is also bounded, for all $t\geq 0$. Since, given \eqref{MARINEVESSEL16}, one has $\ddot{V}_i = -2 s_i^T K_i \dot{s}_i$, which shows that $\ddot{V}_i(t)$ is bounded, for all $t\geq 0$. Furthermore, Barbalat's lemma indicates that $\dot{V}_i(t)$ tends to zero, which implies that $s_i(t)\rightarrow0$ as $t\rightarrow\infty$.

Using \eqref{claw1}, \eqref{claw2} and (\ref{compensator2b}) gives
   \begin{eqnarray}\label{MARINEVESSEL19}
   \dot{q}_{i}-E\dot{\hat{v}}_i+\alpha\left(q_{i}-E\hat{v}_i\right) =s_i-\hat{\kappa}_i\rho_i(z_i) E z_i.
  \end{eqnarray}
With $\tilde{q}_i=\left(q_i-E\hat{v}_i\right)$, equation (\ref{MARINEVESSEL19}) can be rewritten as
   \begin{eqnarray}\label{reMARINEVESSEL19}
   &&\dot{\tilde{q}}_{i}+\alpha \tilde{q}_i  = s_i-\hat{\kappa}_i\rho_i(z_i) E z_{i}.
  \end{eqnarray}
Since,  by Lemma \ref{bounded2}, $\lim_{t\rightarrow\infty}z_{i}(t)= 0$, for $i=1,\dots,N$, equation (\ref{reMARINEVESSEL19}) can be viewed as a stable differential equation in $\tilde{q}_i$ with
$\left(s_i(t)-\hat{\kappa}_i(t)\rho_i(z_i(t)) E z_{i}(t)\right)$ as the input, which are bounded over $t\geq 0$ and tend to zero as $t\rightarrow \infty$. Thus, we conclude that both $\tilde{q}_i(t)=\left(q_i(t)-E\hat{v}_i(t)\right)$ and $\dot{\tilde{q}}_i(t)=\big(\dot{q}_i(t)-E\dot{\hat{v}}_i(t)\big)$ are bounded over $t\geq 0$ and will tend to zero as $t\rightarrow \infty$. These facts together with (\ref{conzero}) and (\ref{dotconzero}) complete the proof.
\end{proof}

\section{A simulation example}\label{section5}

Consider a group of six Euler-Lagrange systems connected according to the communication topology shown in Fig.~\ref{fig1}.
\begin{figure}[htbp]
\centering
\includegraphics[width=0.40\textwidth,trim=200 502 207 128,clip]{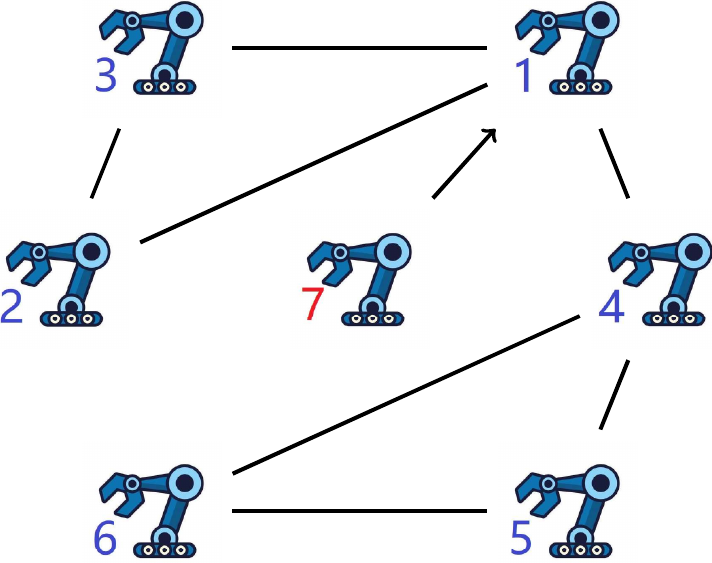}
\caption{ {\myr Communication topology $\bar{\mathcal{G}}$\protect\footnotemark[1]}}\label{fig1}
\end{figure}
\footnotetext[1]{The picture of robot arm is adopted from www.vector4free.com}
It is clear that the communication graph contains a spanning {tree} with the leader (i.e., node 7) being the root. Each follower (i.e., nodes 1-6) is a two-link planar elbow arm whose dynamics are described by \eqref{MARINEVESSEL1},
where $q_i=\col\left(q_{i1}, q_{i2}\right)$ is the joint variable representing the angles of both joints of the two-link robot arm,
\begin{align*}
  M_i\left(q_i\right) &= \left[
                            \begin{array}{cc}
                              \theta_{i1} +\theta_{i2}+2\theta_{i3}\cos q_{i2} &  \theta_{i2}+\theta_{i3}\cos q_{i2} \\
                               \theta_{i2}+ \theta_{i3}\cos q_{i2}&  \theta_{i2} \\
                            \end{array}
                          \right], \\
  C_i\left(q_i,\dot{q}_i\right) &= \left[
                                      \begin{array}{cc}
                                        -\theta_{i3}\dot{q}_{i2}\sin q_{i2}& -\theta_{i3}\left(\dot{q}_{i1}+\dot{q}_{i2}\right)\sin q_{i2} \\
                                         \theta_{i3}\dot{q}_{i1}\sin q_{i2} & 0  \\
                                      \end{array}
                                    \right],\\
 G_i\left(q_i\right) &=\left[
                          \begin{array}{c}
                            \theta_{i4}g\cos q_{i1}+\theta_{i5}g\cos\left( q_{i1}+ q_{i2}\right) \\
                            \theta_{i5}g\cos\left( q_{i1}+ q_{i2}\right) \\
                          \end{array}
                        \right],
\end{align*}
and the unknown vector $\theta_i=\col\left(\theta_{i1},\theta_{i2},\theta_{i3},\theta_{i4},\theta_{i5}\right)$ is induced by the unknown link masses and link offsets \citep*[p.~117]{lewis2003control}.

The leader's signal is generated by a Van der Pol system in the form \eqref{Van0} with%
\begin{align*}
  \frac{\partial p(v, \omega)}{ \partial \omega} &=\left[
                     \begin{array}{ccc}
                        v_2 &  0 &  0 \\
                        0 &  -v_1 & \left(1-v_{1}^2\right)v_2  \\
                     \end{array}
                   \right], \\
  \omega&= \col\left(a,b,c\right),
\end{align*}
where $a$, $b$ and $c$ are {\myr unknown} positive constant scalars. The {states of the} Van der Pol system will converge to a stable limit cycle \citep*{r18}. Therefore, $v(t)$ is bounded for all $t\geq 0$ and for any initial condition $v(0)$.

{\myr It is noticed that $p(x,\omega)$ is in the polynomial differential form of degree $3$. Then, for any $x\in \mathds{R}^{2}$, we have
\begin{align*}
\frac{\partial p(x,\omega)}{\partial x}
=\left[
             \begin{array}{c}
             \frac{\partial p_1(x,\omega)}{\partial x} \\
             \frac{\partial p_2(x,\omega)}{\partial x}
             \end{array}
           \right]
=\left[
             \begin{array}{c;{1pt/2pt} c}
               0 & a\\
               \hdashline[1pt/2pt]
               -b-2c x_{1}x_2&  c(1-x_1^2)\\
             \end{array}
           \right].
\end{align*}
Hence, for some unknown {\myr sufficiently large} positive constant $\bar{\omega}$, we have
\begin{align*}
\sum\nolimits_{l=1}^{2}\Big\|\frac{\partial p_l(x,\omega)}{\partial x}\Big\|\leq & {\myr a+\|b+2c x_{1}x_2\|+\|c(1-x_1^2)\|}\\
\leq & {\myr a+b +c+2c\|x_{1}\|\|x_2\|+c\|x_1\|^2}\\
\leq & 2\bar{\omega}+{2}\bar{\omega}\|x\|^2.
\end{align*}
Then, we have
\begin{align*}
\sum\nolimits_{l=1}^{2}\int_{0}^{1} \Big\|\frac{\partial p_l(x,\omega)}{\partial x}\Big\|_{x=v+\vartheta_l\tilde{v}_i}d\vartheta_l
\leq & {\myr 2\bar{\omega}[1+(\|v\|+\|\tilde{v}_i\|)^2]}\\
\leq & \sum\nolimits_{k=0}^{2}c_k(\omega,v)\|\tilde{v}_i\|^k.
\end{align*}
where $c_0(\omega,v)=2\bar{\omega}(1+\|v\|^2)$, $c_1(\omega,v)=4\bar{\omega}\|v\|$ and $c_2(\omega,v)=1$.
 Then, by Lemma \ref{lemmarho}, we can choose
\begin{align}
\rho_i(z_i)=&2+6(\|z_{i}\|+\|z_{i}\|^{2}+\|z_{i}\|^{3}+\|z_{i}\|^{4}).
\end{align}
Please be noted that $b_i=2$ and $a_i=6$ in $\rho_i(z_i)$ are randomly chosen in $\mathds{R}^{+}$.}
By Lemma \ref{lemmarho}, for $i=1,\dots,6$, we can design a learning-based distributed observer for (\ref{leader}) in the form \eqref{compensator2},
where $\mu=10$, and $\hat{\omega}_{i}=\col\big(\hat{a}_i,\hat{b}_i,\hat{c}_i\big) \in \mathds{R}^3$ is the estimation of $\omega$.

Based on this observer, we can further synthesize a control law of the form (\ref{dstureq11i}) with the following parameters: $K_i=20I_2$, $\alpha=2$, and $\Gamma_i=10$. For $i=1,\dots,N$, the actual values of $\theta_i$ are given as follows:
\begin{eqnarray*}
 && \theta_1 = \col(0.64,1.10, 0.08,0.64,0.32), \\
  && \theta_2 =\col(0.76,1.17,0.14,0.93,0.44),\\
  &&\theta_3 = \col(0.91,1.26,0.22, 1.27,0.58),\\
  &&\theta_4 = \col(1.10,1.36,0.32,1.67,0.73),\\
  &&\theta_5 = \col(1.21,1.16,0.12,1.45, 1.03),\\
   &&\theta_6 = \col(1.31,1.56,0.22,1.65,1.33).
\end{eqnarray*}

\begin{figure}[htbp]
  \centering
  \epsfig{figure=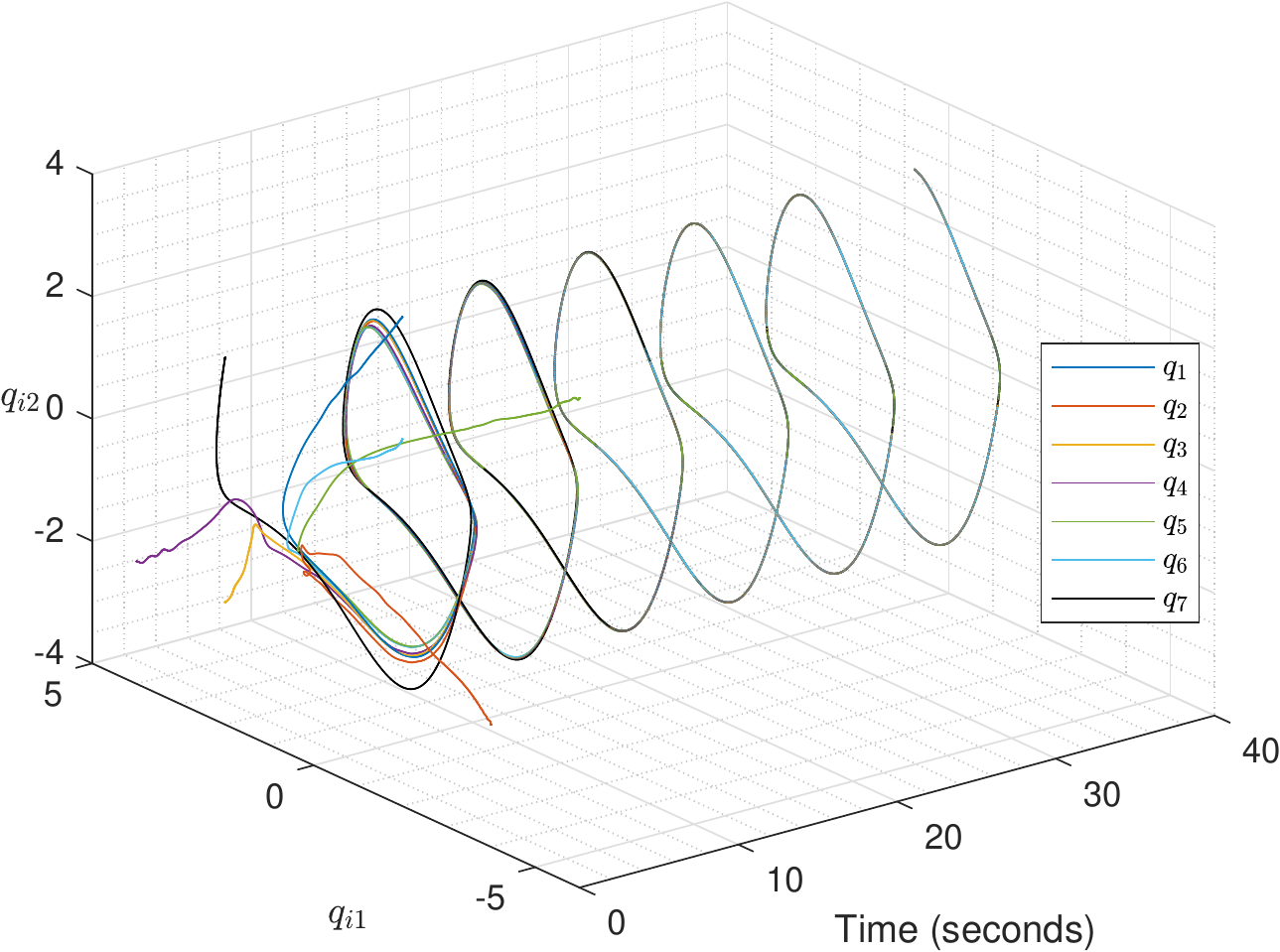,width=\linewidth}
  \caption{Trajectories of $q_i=\col(q_{i1}, q_{i2})$, $i=1,\dots,7.$ }\label{fig2qq}
  \end{figure}
\begin{figure}[htbp]

  \epsfig{figure=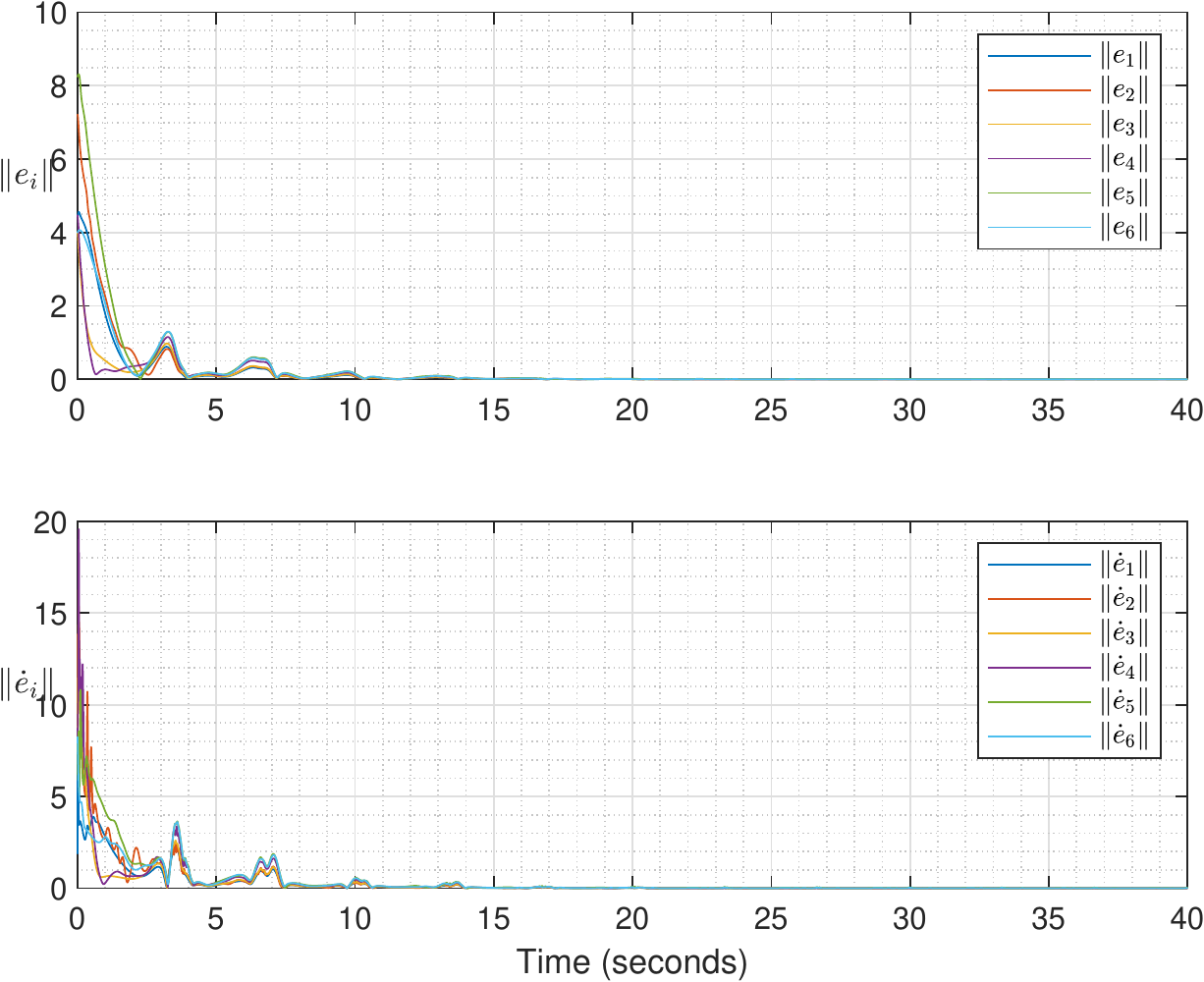,width=\linewidth}
  \caption{Trajectories of $e_i$ and $\dot{e}_i$, $i=1,\dots,6.$}\label{fig2}

 \bigskip
  \centering\setlength{\unitlength}{0.75mm}
  \epsfig{figure=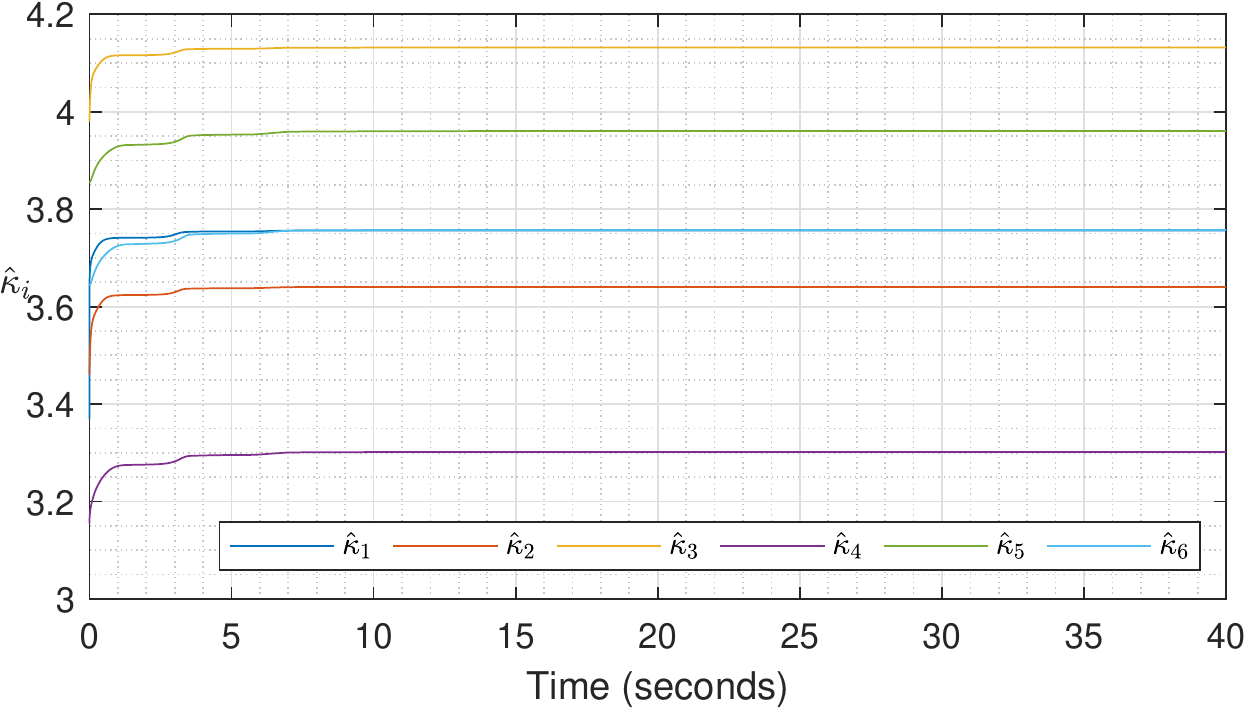,width=\linewidth}

  \caption{Trajectories of $\hat{\kappa}_i$,~ $i = 1, \dots, 6.$}\label{figkapa}
  \bigskip
  \centering\setlength{\unitlength}{0.75mm}
  \epsfig{figure=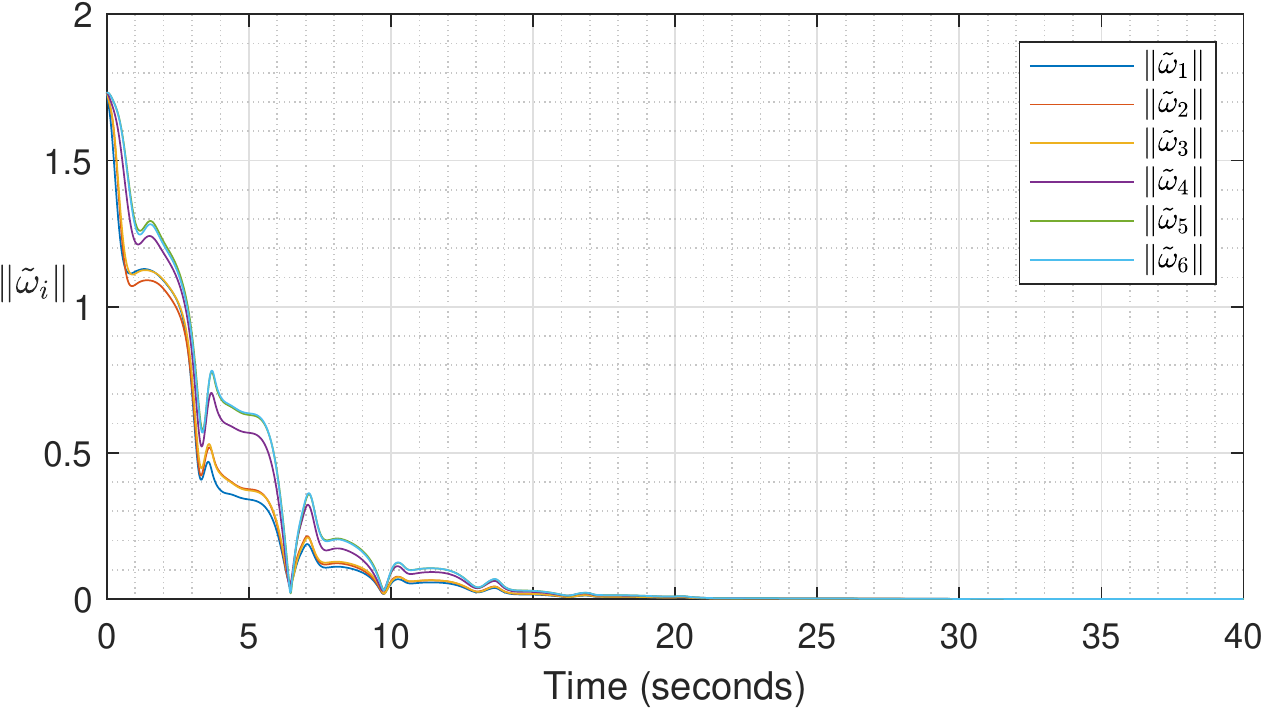,width=\linewidth}
  \caption{Trajectories of $\|\tilde{\omega}_i\|$,~ $i = 1, \dots, 6.$}\label{fig5}
\end{figure}
Simulation is conducted with the following initial conditions: $q_{1}(0)=\col(-1,2)$, $q_{2}(0)=\col(-2,-1)$, $q_{3}(0)=\col(1,-1)$, $q_{4}(0)=\col(2,-1)$, $q_{5}(0)=\col(-3,2)$, $q_{6}(0)=\col(-1,1)$, $\dot{q}_i=0$, $\hat{\theta}_i(0)=0$, $\hat{\kappa}_1(0)=3.3689$, $\hat{\kappa}_2(0)=3.4607$, $\hat{\kappa}_3(0)=3.9816$, $\hat{\kappa}_4(0)=3.1564$, $\hat{\kappa}_5(0)=3.8555$, $\hat{\kappa}_6(0)=3.6448$, $\hat{\omega}_i(0)=0$, $i=1,\dots,6$. The leader's initial condition $v(0)=\col(2,2)$ is chosen in  $\mathds{V}_0=\{v(0)|\|v(0)\|\leq 8\}$, and the actual unknown constants are $a = 1$, $b= 1$ and $c=1$.
Figure \ref{fig2qq} shows that both joint angles $q_i=\col(q_{i1},q_{i2})$ of all six robot arms asymptotically achieve synchronization with the output $q_{N+1}$ of the leader system. The synchronization errors $e_i=q_i - q_7$ and $\dot{e}_i=\dot{q}_i - \dot{q}_7$ of both joint angles and joint angular velocities vanishes as $t \to \infty$, i.e., $||e_i(t)|| \to 0$ and $||\dot{e}_i(t)|| \to 0$ (see Fig. \ref{fig2}). Moreover, the adaptive parameters $\hat{\kappa}_i$ eventually converge to some constants (see Fig. \ref{figkapa}).
%
 By Remark \ref{remwu}, $\phi^T(v)$ is {persistently exciting} for any $v(0)\neq0$. Thus,
  the distributed observer (\ref{compensator2}) will guarantee $\tilde{\omega}_i(t)\rightarrow0$ as $t\rightarrow\infty$, for $i=1,\dots,6$ (see Fig. \ref{fig5}).

\begin{rem}\myr In this  example, each follower agent does not know any knowledge of the leader. It is hard to find a time-varying bound coupled by the network and nonlinear leader, and thus one cannot find a sufficiently large  smooth function to design the nonlinear distributed observer as in \cite{DongChen202021}. Furthermore, since the leader is described by a nonlinear dynamics, those designs proposed in \cite{wanghuang2018} and \cite{baldi2020distributed} for linear leaders are not applicable in our case. More importantly, the solution of the closed-loop system may be unbounded in our case, therefore, the global Lipschitz condition \citep*{wangslotine2004,zhou2006adaptive,yu2009pinning}, the quadratic condition \citep*{lu2004synchronization, delellis2010synchronization}, and the contracting condition \citep*{wang2006theoretical} cannot be used to analyze the convergence of our learning-based distributed observer.
\end{rem}

\section{Conclusion}\label{section6}
This paper  studied the leader-following synchronization problem of heterogeneous multi-agent systems subject to a class of uncertain nonlinear leader systems.
For this purpose, we first established a learning-based distributed observer for the nonlinear leader system whose parameters are unknown. This observer can globally estimate and pass the leader's state to each follower through the communication network without knowing the leader's parameters. Further, under certain persistent excitation condition, this observer can also asymptotically learn the unknown parameters of the leader' system. Some common assumptions in the analysis of the synchronization problem, such as the bounded Jacobian matrix assumption, the quadratic condition, and the known parameters assumption, are removed.
%
Based on this distributed observer, finally, we synthesized an adaptive distributed control law for solving the leader-following synchronization problem of multiple heterogeneous Euler-Lagrange systems via the certainty equivalence principle.
{\myr
The Lyapunov function candidate
\eqref{reeq3b} relies heavily on the
symmetry property of the Laplacian matrix of the communication graph of the undirected graph. It turns out to be very challenging in extending Theorem \ref{bounded2} to the case of directed graphs.
In
the future, we will consider developing a learning-based
distributed observer for multi-agent systems with an uncertain leader over
general directed graphs.
}

%
\bibliographystyle{ifacconf}
\bibliography{myref}
\end{sloppypar}
\end{document}